\documentclass[runningheads,a4paper]{llncs}

\usepackage{amsfonts,amsmath,amssymb,subfigure}
\usepackage{dsfont}

\usepackage{graphicx}
\usepackage{color}
\usepackage{latexsym}
\usepackage{bm}
\usepackage{enumerate,cite}

\newtheorem{invariant}{Invariant}
\newtheorem{observation}{Observation}

\newcommand{\cvec}{{\bf c}}

\newcommand{\dvec}{\bm{\delta}}
\newcommand{\pvec}{\bm{\phi}}

\newcommand{\head}{\ensuremath{\mathit head}}
\newcommand{\tail}{\ensuremath{\mathit tail}}

\newcommand{\rev}[1]{\ensuremath{\mbox{\it rev}\,(#1)}}
\newcommand{\leftd}[1]{\ensuremath{\text{left}\,(#1)}}
\newcommand{\rightd}[1]{\ensuremath{\text{right}\,(#1)}}

\newcommand{\xty}[2]{\ensuremath{#1}-to-\ensuremath{#2}}
\newcommand{\pstart}[1]{\ensuremath{\mathit{start}(#1)}}
\newcommand{\pend}[1]{\ensuremath{\mathit{end}(#1)}}

\newcommand{\ccw}{c.c.w.\ }
\newcommand{\cw}{c.w.\ }

\newcommand{\wrt}{w.r.t.\ }

\newcommand{\fvec}{{\bf f}}

\usepackage{marvosym}
\def\snip{\mathbin{\raisebox{0.15ex}{\rotatebox[origin=c]{60}{\Rightscissors}\!}}}

\usepackage{url}
\urldef{\mailsa}\path|{glencora, harutyan}@eecs.orst.edu|    
\newcommand{\keywords}[1]{\par\addvspace\baselineskip
\noindent\keywordname\enspace\ignorespaces#1}

\begin{document}

\mainmatter  

\title{Maximum $st$-flow in directed planar graphs via shortest paths}

\titlerunning{Maximum $st$-flow in directed planar graphs via shortest paths}

%
%
\author{Glencora Borradaile \and Anna Harutyunyan}
%

 \institute{Oregon State University}

%
%

\maketitle

\begin{abstract}
  Minimum cuts have been closely related to shortest paths in planar
  graphs via planar duality -- so long as the graphs are undirected.
  Even maximum flows are closely related to shortest paths for the
  same reason -- so long as the source and the sink are on a common
  face.  In this paper, we give a correspondence between maximum flows
  and shortest paths via duality in {\em directed} planar graphs with
  no constraints on the source and sink.  We believe this a promising
  avenue for developing algorithms that are more practical than the
  current asymptotically best algorithms for maximum $st$-flow.

  \keywords{maximum flows, shortest paths, planar graphs}
\end{abstract}

\section{Introduction}

The asymptotically best algorithm for max $st$-flow in directed planar
graphs is due to Borradaile and Klein~\cite{BK09}.  The algorithm is
the {\em leftmost augmenting-path algorithm}; in each iteration, flow
is pushed along the leftmost\footnote{Formal definitions are delayed
  until the next section.}  path from $s$ to $t$.  There are at most
$2n$ iterations~\cite{Erickson10} and, using a dynamic-trees data
structure, each iteration can be implemented in $O(\log n)$ time
(e.g.~\cite{ST83}).  The algorithm is, in fact, a generalization of
the algorithm first suggested by Ford and Fulkerson in their seminal
{\em Max-Flow, Min-Cut} paper for the planar case when $s$ and $t$ are
on a common face, or $st$-planar~\cite{FF56}.  However, in the
$st$-planar case, each iteration of the leftmost augmenting-path
algorithm can be implemented in $O(\log n)$ time using priority
queues~\cite{Hassin81}.

Priority queues are arguably simpler and more practical than dynamic
trees.  In fact, Tarjan and Werneck have shown experimentally that a
na\"ive, linear-time-per-operation implementation of dynamic-trees
outperforms more sophisticated logarithmic-time implementations in
many scenarios~\cite{TW10}.  The reason that priority queues are
sufficient for the $st$-planar case is because the leftmost-augmenting
paths algorithm, in this case, can be implemented as Dijkstra's
algorithm in the dual graph.  We propose and pursue the following:
\begin{conjecture}\label{conj}
  An augmenting-path algorithm for max $st$-flow in directed planar
  graphs can be implemented with $O(n)$ queries to a priority-queue.
\end{conjecture}
In this paper, we make progress toward this conjecture by showing that
the leftmost augmenting-path algorithm in the general case can be
reduced to the $st$-planar case in a covering graph.

\subsection{Related results}

The relationship between shortest paths, minimum cuts and maximum
flows in planar graphs has been studied from very early in the history
of maximum flow algorithms. Whitney showed that a minimum $st$-cut
corresponds to a shortest $st$-separating cycle in the graph's
dual~\cite{Whitney1933}.  Hassin showed that for $st$-planar graphs,
by splitting the face common to $s$ and $t$ in two (becoming two
vertices in the dual), finding a shortest-path tree in the dual
(rooted at one of the new vertices\footnote{For details of Hassin's
  algorithm, see Appendix~\ref{app:max-flow-shortest}.}) provides all
the information needed for the maximum $st$-flow~\cite{Hassin81}.

When $s$ and $t$ are not on the same face, the correspondence between
shortest paths and maximum flows is weaker.  In fact, the known
relationships are indirect, reducing max flow or min cut to a {\em
  sequence} of $O(n)$ or $O(\phi)$ $st$-planar max flow or min cut
computations.  Reif showed that when the graph is undirected, there is
a set of at most $n$ nesting, separating cycles in the dual, one of
which is the minimum $st$-separating cycle~\cite{reif83}; these cycles
can be found via divide and conquer, each with a shortest-path
computation.  Hassin and Johnson showed that the shortest-path
distances computed in Reif's algorithm are sufficient to reconstruct a
maximum flow~\cite{HJ85}.  Kaplan and Nussbaum reduced the number of
cycles required for Reif's algorithm to at most $\phi$, where $\phi$
is the minimum number of faces that any $s$-to-$t$ curve must pass
through~\cite{KN11}.  Unfortunately, the $\phi$ shortest-path tree
distances computed are insufficient for reconstructing the flow.  The
only result for maximum $st$-flows or minimum $st$-cuts in directed
planar graphs that utilizes shortest-path computations in the dual, to
our knowledge, is an approach used by Itai and Shiloach~\cite{IS79}
and Johnson and Venkatesan~\cite{JV83} which infeasibly routes flow
along an artificial path of length $\phi$ and reduces the problem to
$\phi$, sequentially-computed $st$-planar maximum flows.  (A summary
of these results is provided in Table~\ref{tab:history} in
Appendix~\ref{app:table}.)

\subsection{Overview}

In Section~\ref{sec:universal-cover}, we define a covering graph by
cutting open a cylindrical embedding of the graph along an $s$-to-$t$
curve and pasting copies of the resulting plane together.  This
covering graph is similar to that used by Erickson~\cite{Erickson10}
for presenting a simpler proof of the bound for Borradaile and Klein's
algorithm.  In Section~\ref{sec:upperm-paths-univ}, we show that the
leftmost flow in the covering graph that contains $O(\phi)$ copies of
the original graph {\em contains} the leftmost flow in the original
graph.\footnote{We need to first find the value of the flow, but
  finding the value of the flow uses the same techniques.}

In terms of running time, our algorithm does not improve on that of
Johnson and Venkatesan's result mentioned above~\cite{JV83}.  Further,
the space requirement for implementing our algorithm would be an
$O(\phi)$ factor higher.  However, the algorithmic technique behind
our algorithm is amenable to the implementation of an augmenting-path
algorithm using only priority queues (and not dynamic trees).  We
discuss this more in Section~\ref{sec:discuss} and provide more
technical evidence for this avenue of research in
Appendices~\ref{app:convergence} and~\ref{sec:algo}.  The algorithms
in these appendices do not suffer from the additional space
requirement.  In fact, the algorithm in Appendix~\ref{sec:algo} gives
an implementation of Borradaile and Klein's algorithm which results in
at most $(2\phi+2)n$ priority queue queries (or at most $2\phi+2$
shortest-path computations).  Unlike other algorithms depending on the
parameter $\phi$, this algorithm does not need to know the value of
$\phi$, and may use many fewer iterations.

We cannot fail to mention that shortest paths (and so, $st$-planar
flows) can be computed in $O(n)$ time~\cite{HKRS97}.  However, this
linear-time algorithm is very complicated and, to our knowledge, has
never been implemented -- it is very likely that Dijsktra's algorithm
would prove to be more practical in most situations.  Further,
Dijkstra's algorithm is more amenable to modifications; in a companion
paper, we use a modification of Dijkstra's algorithm for planar graphs
in a situation where the linear-time algorithm cannot be used.

\subsection{Background} 


The graphs of this paper are directed, but we consider the underlying
undirected graph.  Each edge of the undirected graph is composed of
oppositely-directed {\em darts} each oriented from {\em tail} to {\em
  head}; $\rev \cdot$ maps a dart to its reverse.  Capacities $\cvec$
are non-negative and defined over the set of darts.\footnote{As the
  original graph is directed, $\cvec[d]$ need not be equal to
  $\cvec[\rev{d}]$.  For a directed edge $a$ of the original graph
  with capacity $c$, $\cvec[d] = c$ and $\cvec[\rev{d}] = 0$ where $d$
  is the dart in the underlying undirected graph in the direction of
  $a$. }  Paths and cycles are sequences of darts, and so are
naturally directed.  Paths may visit the same vertex multiple times;
those that do not are {\em simple}.  A path may be trivial in which
case it is a single vertex.  We use $X[a,b]$ to denote subpath of path
or cycle $X$ that include the endpoints $a$ and $b$.  We use $\circ$
to denote the concatenation of paths.

We extend any function or property on elements to sets of elements in
the natural way.

\subsubsection{Flows}

A {\em flow assignment} $\fvec$ is an assignment of real numbers to
the darts of $G$ that is antisymmetric, i.e., $\fvec[d] = -
\fvec[\rev{d}]$.  A {\em pseudoflow} is a flow assignment that {\em
  respects the capacity} of every dart $d$: $\fvec[d]\leq
c[d]$~\cite{Hochbaum08}.  The {\em net flow} of a vertex $v$ is the
sum of the flow of darts whose head is $v$. Vertices with positive net
flow are {\em excess vertices}; vertices with negative net flow are
{\em deficit vertices}.  A pseudoflow with zero net flow at every
vertex is a {\em circulation}.  An {\em $ST$-flow} or simply, a {\em
  flow}, is a pseudoflow whose only deficit vertices are in $S$, the
sources, and only excess vertices are in $T$, the sinks.  We consider
mostly the single-source, single-sink case in which $S = \{s\}$ and $T
= \{t\}$ and which we refer to as an $st$-flow.  The value of a flow
is the sum net flows of its sources and is denoted $|\fvec|$.

Given capacities $\cvec$ and a flow assignment $\fvec$, the residual
capacity of dart $d$ is $\cvec_\fvec[d] = \cvec[d] - \fvec[d]$.  A
dart $d$ is {\em residual} if $\cvec_\fvec[d] > 0$.  We say that a set
of darts is {\em saturated} if one dart in the set is non-residual
(giving saturated path, cycle, etc.).  Residuality is with respect to a set
of capacities; those capacities may be the original capacities or the
residual capacities as indicated in the text.

\begin{quote}
  A flow is maximum if and only if there are no residual
  source-to-sink paths~\cite{FF56}.  

  A pseudoflow is maximum if and only if there are no residual paths
  from a source or excess vertex to a deficit or sink
  vertex~\cite{Hochbaum08}.
\end{quote}

We say that a set of darts are {\em flow darts} if each dart in the
set has positive flow (giving flow path, cycle, etc.).  We say that a
flow assignment is acyclic if there are no flow cycles.  {\em
  Reducing} the flow on a set of flow darts is the operation of
reducing the flow assignment on those darts by a fixed amount (and
changing the flow assignment on the reverse of those darts
accordingly). Acyclic, maximum pseudoflows can be converted into
maximum flows by reducing the flow on source-to-excess and
deficit-to-sink flow paths until all but the sources and sinks are
balanced.  Although we do not use this fact algorithmically, this can
be done in linear time~\cite{JV82,BKMNW11}.

\subsubsection{Planar graphs}  

We use the usual definitions for a planar embedded graph $G$ and its
dual $G^*$.  See Ref.~\cite{Diestel05}, e.g., for formal definitions.



Suppose $a$, $b$, and $d$ are darts such that
$\head(a)=\tail(b)=\head(d)=v$: we say $d$ {\em enters} $a\circ b$ at
$\head(a)$.  If the clockwise ordering of these darts around $v$ in
the embedding is $a,b,d$, then $d$ enters $a \circ b$ from the right
and $\rev{d}$ leaves $a \circ b$ from the right.  Entering and leaving
from the left are defined analogously.
We say that $P$ {\em crosses} $Q$ {\em from right to left} at $X$ if
$X$ is a maximal subpath\footnote{For a trivial path $X$ that is a
  single vertex $x$, $\pstart{X}$, $\pend{X}$ and $\rev{X}$ are $x$
  itself.}  of $Q$ such that either $X$ or $\rev{X}$ is a subpath of
$P$ and $P$ enters $Q$ from the right at $\pstart{X}$ and $P$ leaves
$Q$ from the left at $\pend{X}$.  Crossing from left to right is
defined analogously.  By definition, $X$ cannot be a prefix or suffix
of either $P$ or $Q$. If $P$ and $Q$ are paths that do not cross, then
they are {\em non-crossing}.  A path/cycle is \emph{non-self-crossing}
if for every pair $P$ and $Q$ of its subpaths, $P$ does not cross
$Q$. 
The notions of entering and crossing are
illustrated in Appendix~\ref{app:fig}.

We define a set of darts $\leftd{P}$ that are on the left side of a
path $P$.  $\leftd{P}$ is a subset of the darts whose head or tail
(but not both\footnote{Assuming, w.l.o.g, that there are no parallel
  edges.}) is in $P$.  For a dart $d$ whose head or tail is in $P$ but
is not an endpoint of $P$, $d \in \leftd{P}$ if $d$ enters or leaves
$P$ from the left.  If $d$'s head or tail is an endpoint of $P$, we
break the left/right tie arbitrarily but consistently as follows. For
each vertex $x$, we assign an arbitrarily chosen face $f_x$ adjacent
to $x$. Suppose $P$ is an $s$-to-$t$ path.  If the head or tail of $d$
is $s$ (resp.\ $t$), then $d \in \leftd{P}$ if $d$ is in the \cw
(resp.\ \ccw) ordering of the darts and faces around $s$ (resp.\ $t$)
in the embedding between $f_s$ (resp. $f_t$) and the first (resp.\
last) dart of $P$.  $\rightd{P}$ is defined
analogously. 

We define the graph $G\snip P$ as the graph {\em cut open along $P$}.
$G\snip P$ contains two copies of $P$, $P_L$ and $P_R$, so that that
the edges in $\leftd{P}$ are adjacent to $P_R$ and the edges in
$\rightd{P}$ are adjacent to $P_L$.

Finally, we remind the reader of the parameter $\phi$ defined in the introduction and which we will use throughout:  $\phi$ is the {\em minimum number of faces that any curve drawn in the surface in which $G$ is embedded that any $s$-to-$t$ path must go through}.

\subsubsection{Clockwise and leftmost} 

A {\em potential assignment} $\pvec$ is an assignment of real numbers
to the faces of a planar graph.  Corresponding to every circulation in
a planar graph, there is a potential assignment such that the flow on
a dart $d$ is given by the difference between the face on the right
side and left side of $d$. We take the potential of the infinite face
to be 0. A circulation is {\em clockwise} (c.w.) if all the potentials
are non-negative. A cycle $C$ is {\em clockwise} if the circulation
that assigns +1 to every dart in $C$ and -1 to the reverse of these
darts (and 0 otherwise) is clockwise. The definition of clockwise
depends on the choice of $f_\infty$. Throughout this paper, we will
take $f_\infty = f_t$.

An $s$-to-$t$ path $A$ is {\em left of} an $s$-to-$t$ path $B$ if $A
\cup \rev{B}$ is a clockwise circulation. A path is {\em leftmost} if
there are no paths left of it.  A simple leftmost path from $s$-to-$t$
can be found by a depth-first left-most search~\cite{BK09}.  A flow
assignment is {\em leftmost} if every clockwise cycle is non-residual.
Leftmost paths and circulations are unique~\cite{KNK93}.  The leftmost
$st$-flow of a given value is unique for the same reason.  The next
lemma follows immediately from the definitions:
\begin{lemma}\label{lem:circ-to-cycles}
  A leftmost circulation can be decomposed into a set of
  flow-carrying clockwise simple cycles.
\end{lemma}
Khuller, Naor and Klein illustrated that the leftmost circulation
corresponds to the potential assignment given by shortest-path
distances in the dual (interpreting capacities as
lengths)~\cite{KNK93}.  It follows immediately that the maximum flow
given by Hassin's algorithm, from dual shortest-path distances in the
dual, is leftmost.

We say that (residual) capacities $\cvec$ are \cw acyclic if there are
no \cw cycles that are residual.  In our algorithms, we start with the
residual capacities obtained via dual shortest paths according to
Khuller, Naor and Klein; these are \cw acyclic.  Leftmost flow
assignments \wrt \cw acyclic capacities are acyclic~\cite{BK09}.

\begin{theorem}\label{thm:no-flow-from-left}
  Let $L$ be the leftmost residual $s$-to-$t$-path in $G$ \wrt \cw
  acyclic capacities $\cvec$.  Let $\fvec$ be any
  $st$-flow\footnote{Not necessarily leftmost.} (of any value).  Then no
  simple \xty{s}{t} flow path crosses $L$ from the left to right.
\end{theorem}

\begin{proof}
  Let $F$ be an $s$-to-$t$ flow path. Suppose for a contradiction that
  $F$ crosses $L$ from left to right at $X$ (a subpath of $L$).  Since
  $F$ and $L$ both start at $s$, $L[s, \pstart X)$ must contain a
  vertex of $F$.  Let $v$ be the last such vertex.  Since $F$ is a
  flow path, $F$ must be residual \wrt $\cvec$ (in order for flow to
  be routed on it).  Since $L$ is the leftmost residual path \wrt
  $\cvec$, no subpath of $F$ can be left of any subpath of $L$
  (between the same endpoints).  It follows that $L[v,\pstart X] \circ
  \rev{F[v,\pstart X]}$ is a simple clockwise cycle.  Further, since
  $F$ crosses $L$ from left-to-right at $X$, $F$ enters the strict
  interior of this cycle. (See figure in Appendix~\ref{app:fig}.) To
  escape (to get to $t$), $F$ must hit a vertex $u$ of $L[v,\pstart
  X]$ (since $F$ is simple).  Then $F[\pstart X,u] \circ L[u,\pstart
  X]$ is a clockwise cycle, contradicting that the initial capacities
  $\cvec$ are \cw acyclic.
\end{proof}

\section{Infinite covers} 
\label{sec:universal-cover}

In our analysis, we use an infinite covering graph of $G$ derived from
the universal cover of a cylinder on which we embed
$G$~\cite{Spanier94}.  In our algorithm, we use a finite subgraph of
this infinite covering graph.
The construction of our cover is similar to the one
described in Section 2.4 of Erickson's analysis~\cite{Erickson10} of
the leftmost-path algorithm of Borradaile and Klein.  Erickson's cover
was of the dual graph, whereas we remain in the primal as our analysis
remains entirely in the primal.

Embed $G$ on a sphere and remove the interiors of $f_t$ and $f_s$.
The resulting surface is a cylinder with $t$ and $s$ embedded on
opposite ends.  The universal cover of this cylinder is an infinite
strip. The repeated drawing of $G$ on the universal cover of the
cylinder defines a covering graph $\cal G$ of $G$.  For a subgraph $X$
of $G$, we denote the subgraph of $\cal G$ whose vertices and darts
map to $X$ by ${\cal G}[X]$.  We say $\bar X \subset {\cal G}[X]$ is a
{\em copy} of $X$ if $\bar X$ maps bijectively to $X$.  We say that
$\bar X$ is an {\em isomorphic} copy of $X$ if $\bar X$ is isomorphic
to $X$.  Note that an isomorphic copy need not exist. (For example,
there is no isomorphic copy of $G$ in $\cal G$, if $G$ is not $st$-planar.)

For ease of notation, we take the top (resp. bottom) of the cylinder
to correspond to $f_t$ (resp.\ $f_s$); all the vertices in ${\cal
  G}[t]$ (resp.\ ${\cal G}[s]$) are on the top (resp.\ bottom) of the
infinite strip.  We can therefore refer to the left and right {\em
  ends} of the strip (which extend to infinity). We number the copies
of $t$ from left to right: ${\cal G}[t] = \{\ldots t^{-1}, t^0, t^{1},
\ldots\}$, picking $t^0$ arbitrarily.  For a simple $s$-to-$t$ path
$P$ in $G$, we denote by $P^i$ the isomorphic copy of $P$ in ${\cal
  G}[P]$ that ends at $t^i$.  $P^i$ divides the infinite strip into
two sides which are the components of ${\cal G} \snip P^i$, the
portion containing the left end and the portion containing the right
end.  Therefore, we can order the copies from left to right, i.e.,
${\cal G}[P] = \{\ldots, P^{-1}, P^0, P^1, \ldots\}$ where $P^i$ is in
the left component of ${\cal G} \snip P^j$ for all $i < j$.

We further define copies of $G$ \wrt $P$: $G^i_P\cup P^{i+1}$ is the finite component of 
subgraph of $\cal G \snip P^i \snip P^{i+1}$.
\begin{observation} \label{obs:iso}
  $G^i_P \cup P^{i+1}$ is isomorphic to $G \snip P$ with $P^i$ mapped to $P_L$ and $P^{i+1}$ mapped to $P_R$.
\end{observation}
We also number the copies of $s$ from left to right: ${\cal G}[s] =
\{\ldots s^{-1}, s^0, s^{1}, \ldots\}$, picking $s^0$ arbitrarily.
(In the next section, for convenience of notation, we will make this
choice linked to the choice of $t^0$.)

The following lemma connects the relative topology of two paths in $G$
to that in $\cal G$.

\begin{lemma}\label{lem:mapping} 
  Let $P$ be a simple $s$-to-$t$ path and let $Q$ be a simple path
  sharing only its endpoints with $P$.  Consider an isomorphic copy
  $\bar Q$ of $Q$ in ${\cal G}$.  Let $a$ and $b$ be the first and
  last darts of $Q$, respectively.
  \begin{enumerate}
  \item If $a, b \in \leftd{P}$, $\bar Q$ starts and ends on $P^{j}$ for
    some $j$.
  \item If $a, b \in \rightd{P}$, $\bar Q$ starts and ends on $P^{j}$ for
    some $j$.
  \item If $a\in \leftd{P}$, $b\in \rightd{P}$, $\bar Q$ starts on
    $P^{j+1}$ and ends on $P^j$ for some $j$.
  \item If $a\in \rightd{P}$, $b\in \leftd{P}$, $\bar Q$ starts on
    $P^j$ and ends on $P^{j+1}$ for some $j$.
  \end{enumerate}
\end{lemma}

\begin{proof} 
  Consider $Q$ in $G \snip P$.  Since $Q$ shares only its endpoints
  with $P$, $Q$ is path in $G \snip P$.  For the 4 cases of the lemma,
  by definition of $\snip$, $Q$ is a (1) $P_R$-to-$P_R$, (2)
  $P_L$-to-$P_L$, (3) $P_R$-to-$P_L$, and (4) $P_L$-to-$P_R$ path.  The lemma then follows from Observation~\ref{obs:iso}. \qed
\end{proof}

The following two lemmas relate clockwise cycles in $\cal G$ and in
the original graph $G$. The following holds for any simple $s$-to-$t$ path $P$.

\begin{lemma}
\label{lem:maps-of-clockwise-cycles}
Let $u^i$ and $u^j$ be copies in $\cal G$ of the vertex $u$ in $G$
such that $u^i \in G^i_P$ and $u^j \in G^j_P$. Let $Q$ be any
\xty{u^i}{u^j} path in $\cal G$.  If $i < j$, then the mapping of $Q$
into $G$ contains a clockwise cycle.
\end{lemma}

\begin{proof} 
  Suppose $i < j$ and consider the composition of copies of $Q$ to
  create an infinite path $\cal Q$ in $\cal G$.  $\cal Q$ is a
  $-\infty$-to-$\infty$ path and may not be simple.  If we imagine a
  source at $-\infty$ and a sink at $\infty$ and we send one unit of
  flow along $\cal Q$, we have a flow of value 1.  It follows that
  there is a simple $-\infty$-to-$\infty$ path $\cal R$ that contains
  a subset of the darts of $\cal Q$.

  Consider the potential function $\pvec_{\cal G}$ over the faces of
  $\cal G$ that assigns $1$ to every face below $\cal R$ and $0$ to
  all other faces.  Define the potential function $\pvec_G$ over the
  faces of $G$ where $\pvec_G[f] = \max_{\bar f \in {\cal G}[f]}
  {\pvec}_{\cal G}[\bar f]$ for face $f$ of $G$.  By definition of \cw
  circulation, ${\pvec}_G$ defines a clockwise circulation (for
  $f_\infty = f_t$).  By definition of $\pvec_G$, only darts mapped to
  by $\cal R$ are flow darts.  By Lemma~\ref{lem:circ-to-cycles},
  there is a \cw flow cycle in the circulation. \qed
\end{proof}

The proof of the following lemma is very similar to that for
Lemma~\ref{lem:mapping} and omitted in the efforts of brevity:
\begin{lemma}\label{lem:cwmap}
  Let $C$ be a simple clockwise cycle in $\cal G$.  A subset of the
  darts of $C$ maps to a clockwise cycle in $G$.
\end{lemma}


The following lemma is key in bounding the size of the finite portion
of $\cal G$ that we will need to consider in our algorithms.

\begin{lemma}[Pigeonhole]
  \label{lem:pigeonhole}
  Let $P$ be a simple path in $G$.  Let $\bar P$ be an isomorphic copy
  of $P$ in $\cal G$.  $\bar P$ contains a dart of at most $\phi+2$
  copies of $G$ in ${\cal G}$.  If $P$ may only use $s$ and $t$ as
  endpoints, then $\bar P$ contains darts in at most $\phi$ copies of
  $G$ in $\cal G$.
\end{lemma}

\begin{proof} 
  Let $\Pi$ be the smallest set of faces that a curve from $s$ to $t$
  drawn on the plane in which $G$ is embedded must cross.  Embed in
  $G$ an artificial edge in each of these faces to connect $s$ to $t$;
  we refer to this path as $\Pi$ as well; $\Pi$ has $\phi+1$ vertices.

  Consider an isomorphic copy $\bar P$ of $P$ in ${\cal G}$.  Let $i$ ($j$) be the
  minimum (maximum) index such that $\bar P$ contains a vertex $u^i$
  ($v^j$) of $G^i_\Pi$ ($G^j_\Pi$). Since
  $\bar P$ must have darts in copies $i, \ldots, j$ of $G$, $\bar P$ must
  contain a vertex of $\Pi^k$ for $k = i+1, \ldots, j$.  If $j-(i+1) >
  \phi+1$, then by the pigeonhole principle, $\bar P$ must visit two
  copies of the same vertex of $\Pi$, contradicting that $P$ is
  simple.  Therefore, $j-i \le \phi+2$, proving the first part of the
  lemma.

  Suppose $P$ may only use $s$ or $t$ as endpoints.  Let $P'$ be the subpath of
  $P$ with $s$ and/or $t$ removed.  Now, defining $i$ and $j$ as
  above, $\bar P'$ visits the two copies of the same vertex if
  $j-(i+1) > \phi-1$ as $\bar P$ does not contain $s$ or $t$.  Therefore
  $j-i \le \phi$, proving the second part of the lemma.\qed
\end{proof}

\section{Maximum flow, shortest paths equivalences}
\label{sec:upperm-paths-univ}

We are ready to illustrate an equivalence between dual shortest paths
and maximum flow.  We assume that our initial capacities $\cvec$ are
\cw acyclic.  We identify a finite portion of $\cal G$ (containing $k$
copies of $G$), ${\cal G}_k$ (Section~\ref{sec:finite-cover}).  In
this finite cover, we compute the leftmost maximum flow: We first
attach a super-source $S$, embedded below the cover, to each source
with a large-capacity arc and a super-sink $T$, embedded above the
cover, to each sink with a large-capacity arc (creating graph ${\cal
  G}_k^{ST}$ and capacities $\cvec^{ST}$).  We can then find the leftmost
max $ST$-flow $\fvec^{ST}$ via Hassin's method.  We show how to
extract from $\fvec^{ST}$ the value of the maximum $st$-flow $|\fvec|$
in $G$ (Section~\ref{sec:finite-cover}).  Given this value, we are
able to change the capacities $\cvec^{ST}$ so that we can extract
$\fvec$ from $\fvec^{ST}$ (Section~\ref{sec:flow}).

This method requires a factor $k$ additional space.  In
Section~\ref{sec:discuss}, we discuss how this additional space
requirement could be removed and how a dual shortest-path algorithm
could be used to simulate an augmenting-paths algorithm in $G$ even
though $s$ and $t$ are not on a common face.  This, we believe, is a
promising avenue for developing a practical algorithm for maximum flow
in planar graphs.

Our proof technique is as follows.  We can, using dual-shortest path
techniques, compute $\fvec^{ST}$, the leftmost max $ST$-flow in ${\cal
  G}^{ST}$.  We wish to find $\fvec$, the leftmost max $st$-flow in
$G$.  We relate $\fvec^{ST}$ to $\fvec$ by first copying $\fvec$ into
${\cal G}^{ST}_k$ and then modifying the resulting flow, without
changing the flow assignment in the central copy of $G$ of ${\cal
  G}^{ST}_k$. We can therefore extract $\fvec$ from the flow
assignment given by $\fvec^{ST}$ in this central copy of
$G$.

\subsection{The finite cover} \label{sec:finite-cover}

Let $L$ be the leftmost residual $s$-to-$t$ path in $G$ and let
$\fvec$ be the leftmost maximum $st$-flow in $G$; since the capacities
are \cw acyclic, $\fvec$ is acyclic.

Let ${\cal G}_k$ be the finite component of ${\cal G} \snip L^0 \snip
L^k$; ${\cal G}_k$ is a finite cover made of $k$ copies of $G$ (plus
an extra copy of $L$), bounded on the left by $L$ and on the right by
the extra copy of $L$.  The sinks of ${\cal G}_k$ are numbered $t^0,
t^1, \ldots, t^k$ from left to right; we (re)number the sources from
left to right, $s^0, s^1, \ldots, s^k$.  We will refer to the $1^{st}$
through $k^{th}$ copies of $G$ in ${\cal G}_k$ according to the
natural left-to-right ordering.

We start by constructing a maximum multi-source, multi-sink maximum
flow in ${\cal G}_k$ ($\fvec_1$), from a maximum pseudoflow $\fvec_0$
(Lemma~\ref{lem:convert-to-pseudoflow}), which is constructed from
$\fvec$ (Lemma~\ref{lem:convert-to-flow}).  This construction will
guarantee that the flow in the central copy of ${\cal G}_k$ is exactly
$\fvec$.  However, $\fvec_1$ is not necessarily leftmost and so, we
cannot necessarily compute it.  We relate $\fvec_1$ to the leftmost
max $ST$-flow in ${\cal G}^{ST}_k$ (which we can compute), in the next
section.

Let $\fvec_0$ be a flow assignment for ${\cal G}_k$ given by
$\fvec_0[\bar d] = \fvec[d],\ \forall \bar d \in {\cal G}[d]$.  We
overload $\cvec$ to represent capacities in both $G$ and ${\cal G}_k$,
where capacities in ${\cal G}_k$ are inherited from $G$ in the natural
way.

\begin{lemma} \label{lem:convert-to-pseudoflow}
  For $k > \phi+2$, $\fvec_0$ is a maximum pseudoflow with excess
  vertices on $L^0$ and deficit vertices on $L^k$.
\end{lemma}

\begin{proof} 
  First notice that $\fvec_0$ is balanced for all vertices in ${\cal
    G}_k$ except those on $L^0$ and $L^k$: they may contain excess and
  deficit vertices. By Lemma~\ref{thm:no-flow-from-left}, there are no
  flow paths in $\fvec$ that push flow from the left of $L$ to the
  right of $L$. It follows that $L^0$ only contains vertices with
  excess or balanced vertices, and $L^k$ only contains vertices with
  deficit or balanced vertices.

  In ${\cal G}_k$ and $G$, we use residual to mean \wrt
  $\cvec_{\fvec_0}$ and $\cvec_\fvec$, respectively.  A pseudoflow is
  maximum if there is no \xty{S\cup V^+}{T\cup V^-} residual path $P$,
  where $V^+$ is the set of excess vertices and $V^-$ is the set of
  deficit vertices. If $P$ is a path from a source to a sink in ${\cal
    G}_k$, then $P$ maps to an \xty{s}{t} path in $G$; therefore, $P$
  cannot be residual in either graph.

  It remains to show that there are no \xty{V^+}{T}, \xty{S}{V^-} or
  \xty{V^+}{V^-} residual paths.  The proof for the first two cases
  are symmetric; we only prove one here.
  
  \paragraph{There are no \xty{V^+}{T} residual paths.} Consider for a
  moment the flow assignment for ${\cal G}$: $\fvec'[\bar d] =
  \fvec[d]$, $\forall \bar d \in {\cal G}[d]$.  For $v^+ \in V^+$ to be
  an excess vertex, there must be a $v$-to-$t$ flow path $Q$ in
  $\fvec$ where $v$ is the vertex in $G$ that $v^+$ maps to.  There is
  a copy $\bar Q$ of $Q$ in $\cal G$ that starts at $v^+$ (where we
  remind the reader that ${\cal G}_k$ is a subgraph of $\cal G$).  By
  Lemma~\ref{thm:no-flow-from-left}, $Q$ does not cross $L$ from left
  to right and so $\bar Q$ is left of $L^0$.

  Now, for a contradiction, let $R$ be a \xty{v^+}{t^i} residual path,
  for some $t^i \in T$.  Since the reverse of a flow path is residual,
  $\rev{Q} \circ R$ is a residual $t^j$-to-$t^i$ path in $\cal G$
  (\wrt $\fvec'$).  Since $\bar Q$ is left of $L^0$, $j \le i$.  If $j
  = 0$, $\bar Q \circ \rev{L^0[v^+,t^0]}$ would be a clockwise cycle,
  which, by Lemma~\ref{lem:cwmap}, would witness a clockwise cycle in
  $G$; since $Q$ is residual \wrt the original capacities $\cvec$,
  this cycle would contradict the leftmost-ness of $L$.  Therefore $j
  < i$.  By Lemma~\ref{lem:maps-of-clockwise-cycles}, $\rev{Q} \circ
  R$ implies the existence of a clockwise residual cycle in $G$,
  contradicting the leftmost-ness of $\fvec$.

  \paragraph{There are no \xty{V^+}{V^-} residual paths.} Let $v^+$
  and $v^-$ be vertices in $V^+$ and $V^-$, respectively. A
  \xty{v^+}{v^-} path $R$ must go through all $k$ copies of $G$ in
  ${\cal G}_k$.  Since $k > \phi+2$, by the Pigeonhole Lemma, $R$
  contains two copies $u^i$ and $u^j$ of the same vertex $u$ in $G$,
  and in fact, must contain a subpath from $u^i$ to $u^j$ for $i < j$.
  By Lemma~\ref{lem:maps-of-clockwise-cycles}, this implies that there
  is a clockwise residual cycle in $G$, contradicting $\fvec$ being
  leftmost. \qed
\end{proof}

\begin{lemma} \label{lem:convert-to-flow} There is a maximum $ST$-flow
  $\fvec_1$ in ${\cal G}_k$ that is obtained from $\fvec_0$ by
  removing flow on darts in the first and last $\phi$ copies of $G$ in
  ${\cal G}_k$.  Further, the amount of flow into sink $t^i$ for $i
  \le k-\phi$ and the amount of flow out of source $s^j$ for $j \ge
  \phi$ is the same in $\fvec_0$ and $\fvec_1$.
\end{lemma}

\begin{proof} 
  Since $\cvec$ are clockwise-acyclic capacities, and $\fvec$ is
  leftmost, $\fvec$ is acyclic.  It follows that $\fvec_0$ is acyclic.
  Since $\fvec_0$ is an acyclic maximum pseudoflow, it can be
  converted to a maximum flow by flow-cancelling
  techniques~\cite{Hochbaum08}; i.e., by removing flow from
  source-to-excess flow paths and deficit-to-sink flow paths.  Let $P$
  be such a flow path.  $P$ maps to a flow path in $G$ and so must map
  to a simple path in $G$.  By the Pigeonhole Lemma, $P$ must be
  contained within $\phi$ copies of $G$.  This proves the first part
  of the lemma.  Since $P$ cannot start at $s^j$ for $j \ge \phi$
  without going through more than $\phi$ copies (and likewise, $P$
  cannot end at $t^i$ for $i \le k-\phi$), the second part of the
  lemma follows. \qed
\end{proof}



\subsection{Value of the maximum flow}\label{sec:value}

In the next lemma, we prove that from $\fvec^{ST}$, the leftmost maximum $ST$-flow in ${\cal G}_k^{ST}$, we can extract $|\fvec|$, the value of the maximum $st$-flow in $G$.

\begin{lemma} \label{lem:value}
  For $k \ge 4\phi$, the amount of flow through $s^{2\phi}$ in $\fvec^{ST}$, the
  leftmost maximum flow in ${\cal G}_k^{ST}$, is $|\fvec|$.
\end{lemma}

\begin{proof}
  We show that the amount of flow leaving $s^{2\phi}$ in $\fvec^{ST}$
  is the same as in $\fvec_1$.  By Lemma~\ref{lem:convert-to-flow},
  the amount of flow leaving $s^{2\phi}$ is the same in $\fvec_1$ as
  $\fvec_0$ which is the same as the amount of flow leaving $s$ in
  $\fvec$; this proves the lemma.
  
  First extend $\fvec_1$ into an $ST$-flow, $\fvec_1^{ST}$, in ${\cal
    G}_k^{ST}$ in the natural way (by setting the flow on the darts
  adjacent to $S$ and $T$ to satisfy the balance constraint).  Since
  $\fvec_1$ is a maximum multi-source, multi-sink flow in ${\cal
    G}_k$, $\fvec_1^{ST}$ is a max $ST$-flow in ${\cal G}_k^{ST}$.

  To convert $\fvec_1^{ST}$ into a leftmost flow, we must saturate the
  clockwise residual cycles; this is done with a \cw circulation,
  which, by Lemma~\ref{lem:circ-to-cycles}, can be converted into a
  set of flow carrying \cw simple cycles.  Suppose for a contradiction
  that one such cycle $C$ changes the amount of flow through
  $s^{2\phi}$.  Since $C$ is a flow cycle in the circulation (which is
  residual \wrt $\cvec_{\fvec_1^{ST}}$), $C$ is residual \wrt
  $\cvec_{\fvec_1^{ST}}$.  If $C$ changes the amount of flow through
  $s^{2\phi}$, $C$ must go through $S$.  $C$ cannot visit $T$, for if
  it did, $C$ would include a source-to-sink path.  This path is also
  residual \wrt $\cvec_{\fvec_1^{ST}}$, contradicting that
  $\fvec_1^{ST}$ is maximum.  Therefore $C$ must contain a
  $s^i$-to-$s^{2\phi}$ path $P$ that is in ${\cal G}_k$; $P$ is
  residual \wrt $\cvec_{\fvec_1}$.  Since $C$ is \cw, $i < 2\phi$.

  We first argue that $P$ must use a dart in the first or last $\phi$
  copies of $G$ in ${\cal G}_k$.  Suppose otherwise.  Then $P$ must
  map to a set $P'$ of darts in $G$ which, by
  Lemma~\ref{lem:convert-to-flow} are residual \wrt $\cvec_\fvec$.  By
  Lemma~\ref{lem:maps-of-clockwise-cycles}, $P'$ contains a clockwise
  cycle, contradicting the leftmostness of $\fvec$.

  Therefore, $P$ must visit a dart in the first or last $\phi$ copies
  of $G$.  It follows that $P$ must cross either from the $\phi^{th}$
  copy to $s^{2\phi}$ or from $s^{2\phi}$ to the $3\phi^{th}+1$ copy
  (possible, since $k \ge 4\phi$).  Then, by the Pigeonhole Lemma ,
  $P$ contains a subpath $Q$ that goes from $\bar v$ to $\bar v'$
  where these vertices are copies of the same vertex that are not in
  the first or last $\phi$ copies and such that $\bar v$ is in an
  earlier copy of $G$ than $\bar v'$.  By
  Lemma~\ref{lem:maps-of-clockwise-cycles}, the map of $Q$ contains a
  clockwise cycle in $G$.  Since $Q$ does not contain darts in the
  first or last $\phi$ copies of $G$, by
  Lemma~\ref{lem:convert-to-flow}, this cycle is residual \wrt $\fvec$
  in $G$, again contradicting that $\fvec$ is leftmost. \qed
\end{proof}

\subsection{Maximum flow}\label{sec:flow}

Now, suppose we know $|\fvec|$ (as per Lemma~\ref{sec:value}).  We
change the capacities of the arcs into $T$ and out of $S$ in ${\cal
  G}_k^{ST}$ to $|\fvec|$.  Call these capacities $\cvec^{|\fvec|}$.
Now, $\fvec_1^{ST}$, as defined in the previous section, respects
$\cvec^{|\fvec|}$ since, by Lemmas~\ref{lem:convert-to-pseudoflow}
and~\ref{lem:convert-to-flow}, the amount of flow leaving any source
or entering any sink in $\fvec_1$ is at most $|\fvec|$.  We prove a
lemma similar to Lemma~\ref{lem:convert-to-flow}.  The proof of this lemma is similar to the proof of Lemma~\ref{lem:value}, so we defer it to Appendix~\ref{app:pf}.

\begin{lemma} \label{lem:convert-to-leftmost} $\fvec_1^{ST}$ can be converted
  into a leftmost maximum $ST$-flow $\fvec^{|\fvec|}$ for the
  capacities $\cvec^{|\fvec|}$ while not changing the flow on darts in
  the first or last $2\phi$ copies of $G$ in ${\cal G}_k$.
\end{lemma}

To summarize, Lemmas~\ref{lem:convert-to-flow}
and~\ref{lem:convert-to-leftmost} guarantee that the maximum leftmost
$ST$-flow, $\fvec^{|\fvec|}$, in ${\cal G}_k^{ST}$ given capacities
$\cvec^{|\fvec|}$ has the same flow assignment on the darts in copy
$2\phi+1$ as $\fvec$ so long as $k \ge 4\phi+1$.  Starting from
scratch, we can find \cw acyclic capacities $\cvec$ via Khuller, Naor
and Klein's method~\cite{KNK93} (one shortest path computation); we
can find $|\fvec|$ (Lemma~\ref{lem:value}, a second shortest path
computation) and then $\fvec$ (Lemma~\ref{lem:convert-to-leftmost}, a
third shortest path computation).  Therefore, finding a maximum
$st$-flow in a directed planar graph $G$ is equivalent to three
shortest path computations: one in $G$ and two in a covering of $G$
that is $4\phi+1$ times larger than $G$.

\section{Discussion} \label{sec:discuss}

Borradaile and Klein's leftmost augmenting-paths algorithm also uses
\cw acyclic capacities as a starting point~\cite{BK09}.  Their
analysis showed that an arc can be augmented and possibly its reverse,
but, when this happens, the arc becomes unusable and will not be a
part of any further augmenting path.  In Appendix~\ref{sec:algo}, we
give an algorithm that implements the leftmost augmenting-paths
algorithm in at most $2\phi+2$ phases.  However, just as the analysis
in Sections~\ref{sec:upperm-paths-univ} doesn't use Borradaile and
Klein's notion of {\em unusability}, neither does the algorithm in
Appendix~\ref{sec:algo}.  In Appendix~\ref{app:convergence} we more
explicitly express the algorithm of
Section~\ref{sec:upperm-paths-univ} as an augmenting paths algorithm
(and remove the extra space requirement).  Since these algorithms are
all variants or direct implementations of the leftmost augmenting
paths algorithm, we should be able to appeal to the unusability of
Borradaile and Klein to tighten the analysis. Combining these ideas
may lead to algorithms that are both asymptotically and practically
superior.
\let\thefootnote\relax\footnotetext{{\em Acknowledgements:} The authors thank J\'{e}remy Barbay for
very helpful discussions.  This material is based upon work supported by
the National Science Foundation under Grant No.\ CCF-0963921.}

\bibliographystyle{plain}
\bibliography{maxflowcover.bbl}

\appendix 
\newpage
\section{History of shortest-path based planar flow algorithms} \label{app:table}

\begin{table}[h]
\centering
\begin{tabular}{|llccl|}
  \hline
  \multicolumn{5}{|c|}{\bf History of planar maximum flow and minimum cut via shortest
    ss} \\ \hline
  {\bf year} & {\bf problem} & {\bf \# SP calls} & {\bf run
    time} &
  {\bf reference}\\
  1969 & min cut, directed $st$-planar & 1 & $\Theta(n)$ & Hu~\cite{Hu} \\
  1979 & max flow, directed $st$-planar & 1 & $\Theta(n)$ & Itai and Shiloach~\cite{IS79}\\
  1983 & max flow, directed  & $\phi$ & $\Theta(n \phi)$ & Johnson and
  Venkatesan~\cite{JV83} \\
  1983 & min cut, undirected  & $\le n$ & $O(n \log n)$ & Reif~\cite{reif83}\\
  1985 & max flow, undirected  & $\le n$ & $O(n \log n)$ & Hassin and
  Johnson~\cite{HJ85}\\
  2011 & min cut, undirected & $\le \phi$ & $O(n \log \phi)$ & Kaplan and
  Nussbaum~\cite{KN11} \\
  2012 & max flow, directed  & $\le 2\phi+2$ & $O(n \phi)$ & [this paper, appendix] \\
  2012 & max flow, directed  & 3 & $\Theta(n \phi)$ & [this
  paper] \\
  \hline
\end{tabular}
\caption{All run times are given using modern techniques, namely the
  linear-time shortest path~\cite{HKRS97} and the
  cycle-cancelling~\cite{KN09} algorithms, that post-date the results
  of last century.  Note that the last algorithm has an $O(\phi)$ factor increase in required space.}
\label{tab:history}
\end{table}

\section{Figures}\label{app:fig}

\begin{minipage}{1.0\linewidth}
    (a) \includegraphics[scale=0.8]{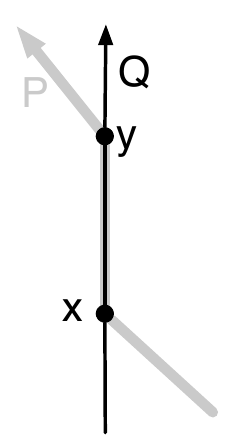}    (b) \includegraphics[scale=0.8]{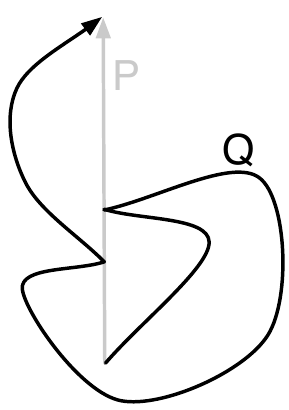}
    (c) \includegraphics[scale=0.5]{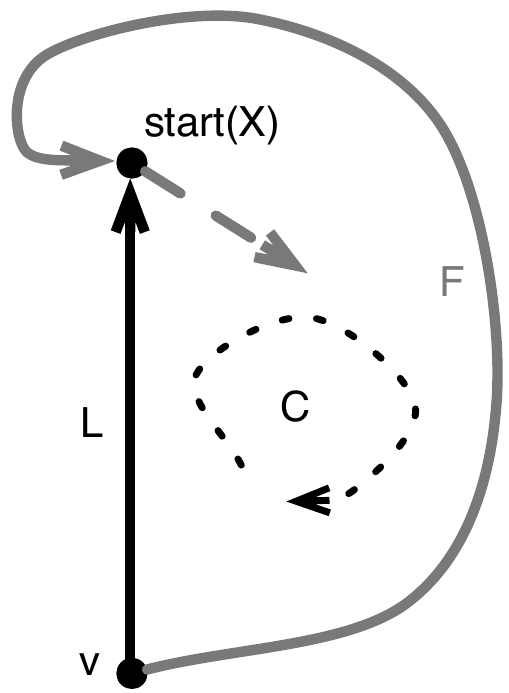}
  \end{minipage}

  \begin{enumerate}[(a)]
  \item $P$ crosses $Q$ from right to left: $P$ enters $Q$ on the right at $x$,
     and $P$ leaves $Q$ on the left at $y$. 
\item  $P$ and $Q$ are
     non-crossing. 
\item An illustration of the construction in the
     proof of Theorem~\ref{thm:no-flow-from-left}.
  \end{enumerate}

\begin{minipage}{1.0\linewidth}
    (a) \includegraphics[scale=0.5]{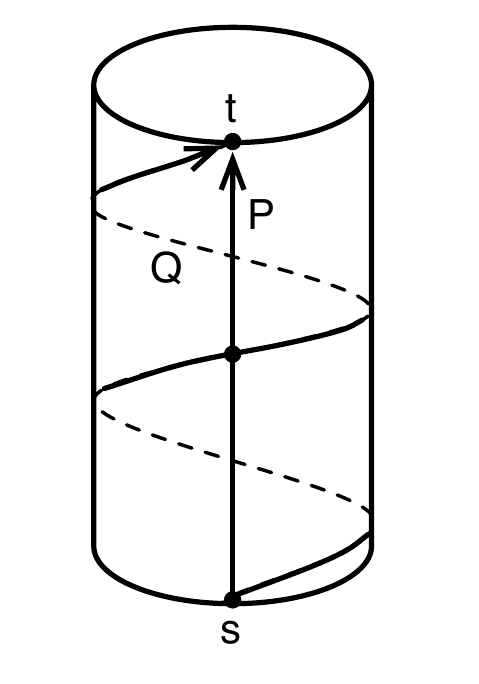}    (b) \includegraphics[scale=0.5]{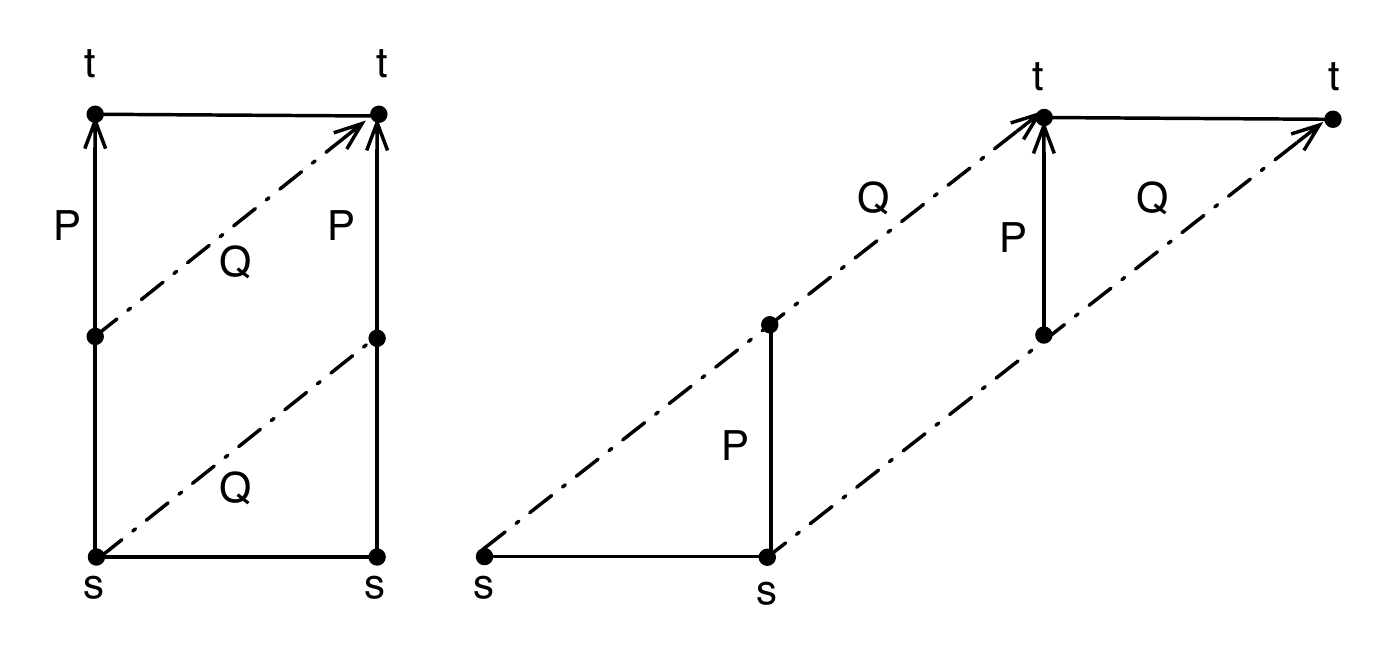}\\
    (c) \includegraphics[scale=0.5]{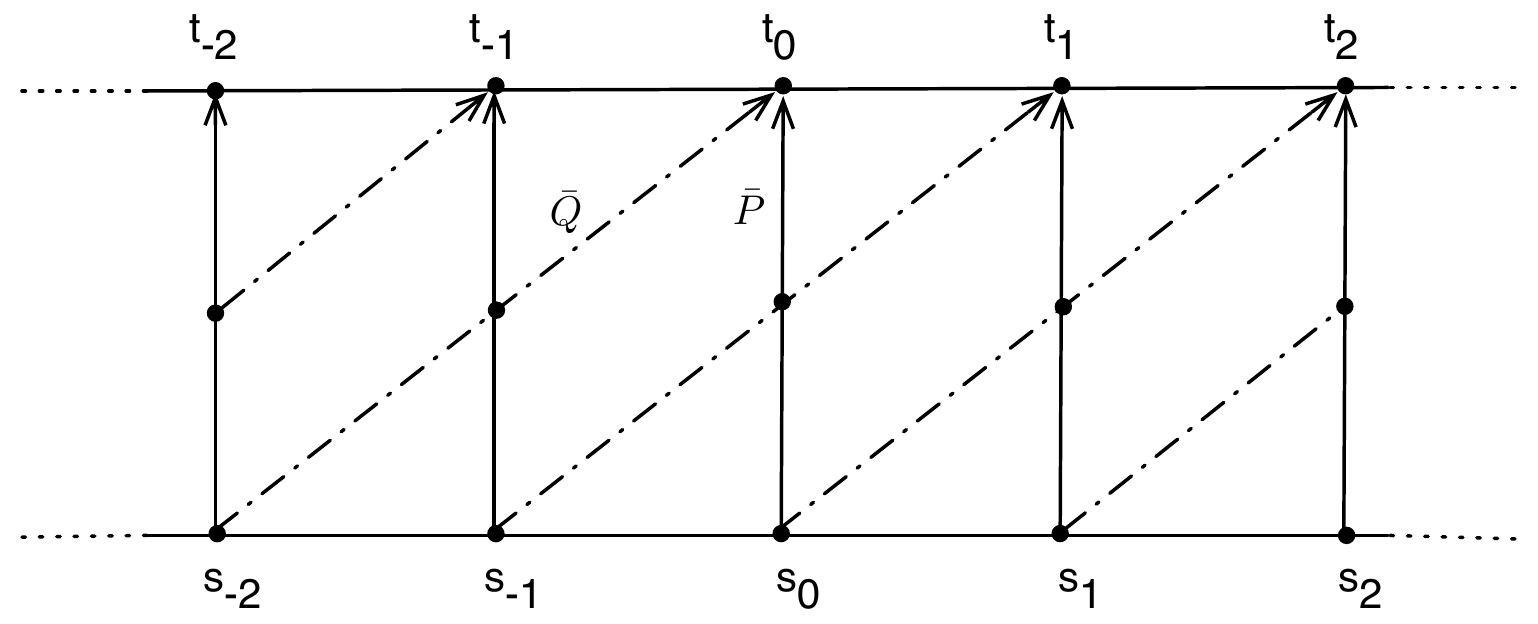}

\end{minipage}

\begin{enumerate}[(a)]
\item $G$ embedded on a cylinder with \xty{s}{t} paths $P$ and
  $Q$.
\item $G \snip P$ and $G \snip Q$.
\item The covering graph $\cal G$ of $G$.
\end{enumerate}

\section{Proof of Lemma~\ref{lem:convert-to-leftmost}} \label{app:pf}

  Let $\cal C$ be the circulation that takes $\fvec_1^{ST}$ to
  $\fvec^{|\fvec|}$. $\cal C$ is leftmost \wrt $\cvec_{\fvec^{ST}_1}$, and by
  Lemma~\ref{lem:circ-to-cycles} can be decomposed into a set of
  clockwise flow cycles. Consider one such cycle $C$. This cycle is
  residual \wrt $\cvec_{\fvec^{ST}_1}$. 


  We first observe that $C$ must visit a vertex in the first $\phi$ copies
  or in the last $\phi$ copies of $G$ in ${\cal G}^{ST}_k$.

  If $C$ goes through $S$, $C$ contains consecutive darts $s^jS,
  Ss^i$.  By the second part of Lemma~\ref{lem:convert-to-flow} and
  the definition of $\fvec_1^{ST}$,
  $Ss^i$ is saturated for all $i \ge \phi$.  Therefore, $i < \phi$.  If
  $C$ goes through $T$, by a similar argument, $C$ must visit a vertex
  in the last $\phi$ copies of $G$.  If $C$ goes through neither $S$
  nor $T$ and if $C$ contains no vertex in the first or last $\phi$
  copies of $G$, then the darts of $C$ are residual in $G$ \wrt $\bf
  f$ by the first part of Lemma~\ref{lem:convert-to-flow}.  By
  Lemma~\ref{lem:cwmap}, $C$ maps to a clockwise cycle in $G$,
  contradicting the leftmost-ness of $\fvec$.

  So, we consider the case when $C$ visits a vertex in the first $\phi$
  copies of $G$; the case for the last $\phi$ copies is symmetric.

  Suppose for a contradiction that $C$ contains a dart $d$ in a copy
  of $G$ greater than $2\phi+1$.  Then $C$ contains a path $P$ that
  starts in copy $\phi$ and ends in copy $2\phi+1$ of $G$.  By the
  first part of Lemma~\ref{lem:convert-to-flow}, the darts of $P$ are
  residual in $\fvec_1$ and so they are also residual in $\fvec$.  By
  the pigeon-hole argument, $P$ contains two copies of the same
  vertex. By Lemma~\ref{lem:cwmap}, $P$ contains a subpath that maps
  to a clockwise cycle in $G$.  This contradicts that $\fvec$ is
  leftmost. \qed

\section{Leftmost maximum flows and shortest paths}
\label{app:max-flow-shortest}

Hassin's algorithm starts by transforming the max flow problem into a
max saturating circulation problem: it introduces a new
infinite-capacity directed arc $ts$
embedded so that every $s$-to-$t$ residual path forms a clockwise
cycle with $ts$ and then saturates the clockwise cycles. This
saturation of clockwise cycles is computed via finding a shortest-paths tree rooted at
$f^*_\infty$ in \footnote{Since we always embed $t$ to be on $f^\infty$ in $G$, $f^*_\infty$
is the head of the dual dart $(ts)^*$ in $G^*$.} the dual graph
$G^*$, interpreting primal capacities as dual distances. Let $\dvec$ be the shortest-path distances in $G^*$ from $f_\infty$. Consider the flow:

\begin{equation}
    \fvec[d] =
    \dvec[\head_{G^*}(d)]-\dvec[\tail_{G^*}(d)] \ \forall \text{ darts }
    d. \label{eq:st-flow}
\end{equation}

After removing the artificial arc $ts$ from $G$, $\fvec$ is a maximum
feasible flow from $s$ to $t$ in $G$, and $\dvec$ is the potential
assignment that induces $\fvec$.

Khuller, Naor and Klein~\cite{KNK93} later showed that a flow derived
from shortest-path distances in the dual in this manner is clockwise acyclic.



\section{Convergence}\label{app:convergence}

We interpret the shortest-paths computation used in Section~\ref{sec:flow}
 as an augmenting paths
algorithm in the original graph $G$, and illustrate that this will
converge in a finite number of augmentations.  

Consider the cover ${\cal G}^{ST}$, as defined in
Section~\ref{sec:upperm-paths-univ}, but with the restriction of $k$
copies removed, and let $\fvec^{ST}$ be the leftmost maximum $ST$-flow
in ${\cal G}^{ST}$. By Lemma~\ref{lem:convert-to-leftmost}, the flow on the darts in copy
$2\phi + 1$ of $G$ in ${\cal G}^{ST}$ corresponds to the leftmost flow in $G$.
If we relax the condition that the capacity of the arcs adjacent to $S$ and $T$ are finite, i.e. if we consider the leftmost
maximum $ST$-flow $\fvec^{ST}$ in ${\cal
  G}^{ST}$ with capacities $\cvec^{ST}$ instead of $\cvec^{|\fvec|}$, we will need to consider a much later copy. 

We show how to convert the flow $\fvec_1^{ST}$ into the leftmost maximum flow $\fvec^{ST}$,
while not changing the flow on darts in copies $2\phi+1$ or later {\em
  except} for saturating $\infty$-to-$T$ paths.

The construction follows much as that the proof of
Lemma~\ref{lem:convert-to-leftmost} except now we may have larger
clockwise cycles through $T$. Such a cycle $C$ would be of the form
$t^iT \circ T t^j \circ P$ where $i < j$. $P$ corresponds to a
counterclockwise residual cycle around $s$ in $G$, and so may be
residual \wrt $\bf f$. Beyond copy $2\phi+1$, all copies of $C$ would
be residual and their union would correspond to a set of
$\infty$-to-$t$-to-$T$ paths. As such, the flow through $t^i$ may be
greater than $|\fvec|$. However, the from-$\infty$ paths will be
saturated at some point: the amount of flow pushed along these paths
can be no greater than $U$, the sum of the capacities of $G$. As this
flow must be directed through some $t^i$ , for $i > 4\phi + U$ (where
the $4\phi$ is inherited from Lemma~\ref{lem:convert-to-leftmost}),
the flow through $t^i$ will be $|\fvec|$. Similarly, the flow in the
copies of $G$ beyond $4\phi+U$ will be the same as $\bf f$ except for
the counterclockwise cycles around $s$ that participate in the
$\infty$-to-$t$ paths. An example of this is illustrated in
Figure~\ref{fig:ccw-cycle-cover}.  

This observation does not seem very useful algorithmically.  However,
consider how Dijkstra's algorithm explores the dual of ${\cal
  G}^{ST}$: when a face is popped off the priority
queue, the wavefront of faces between those that have been popped and
those that have not corresponds to an augmenting path. We can follow
these augmenting paths in the original graph $G$.  These augmenting
paths push flow from $s$ to $t$, but when we are done with a path, we
do not update the capacities, but simply continue.  By the previous
observation, this version of augmenting paths will converge to a
maximum flow (in $O(n(U+\phi))$ iterations).

It is curious to note that although these augmentations correspond to
leftmost augmentations, the flow we end up with may not be leftmost,
as it may saturate a counterclockwise cycle around $s$, as in the
example of Figure~\ref{fig:ccw-cycle-cover}.

\begin{figure}[ht]
  \centering
\subfigure[]
{
\includegraphics[scale=0.6]{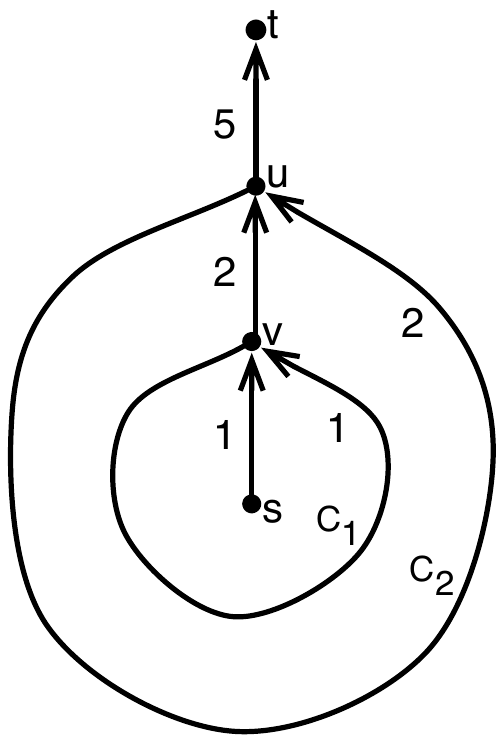}
}
\hspace{1em}
\subfigure[]
{
\includegraphics[scale=0.6]{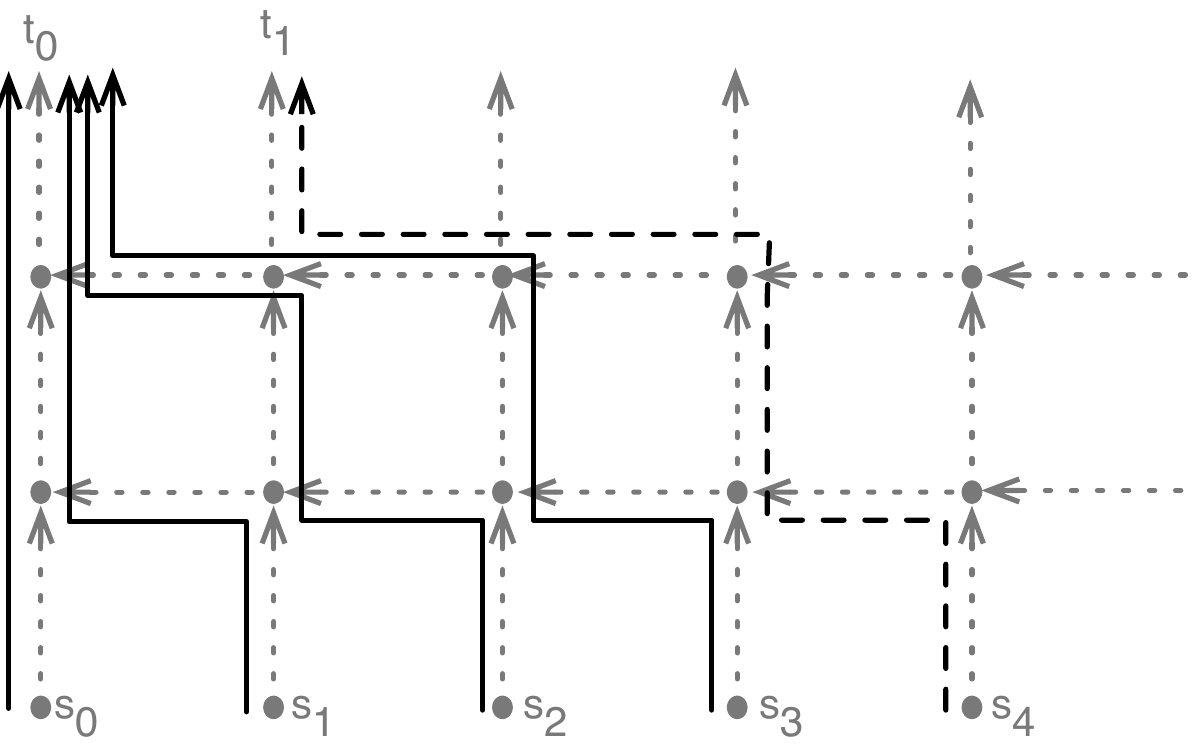}
}

  \caption{(a) $G$ contains two counterclockwise cycles around $s$: $C_1$ and $C_2$, of
    capacities 1 and 2 respectively. Upon routing 1 unit of flow on
    the \xty{s}{t} path, $C_1\circ vt\circ tv$ and $C_2\circ
    ut\circ tu$ are residual cycles through $t$ around $s$. (b) $G$'s
    half-infinite cover is given with grey dotted
    edges. Solid paths correspond to augmentations to $t_0$; the
    dashed path is an augmentation to $t_1$. Convergence
    occurs after $s_3$, instead of $s_0$.}
\label{fig:ccw-cycle-cover}
\end{figure}

\section{Maximum flow via sequential shortest-path computations}
\label{sec:algo}

The algorithm implicit in Section~\ref{sec:flow} both requires the value
of $\phi$ and will take time proportional to $\Theta(\phi)$.  Here we
present a different algorithm based on similar analysis techniques
that is adaptive in the sense that the running time is naturally
upper bounded by a function of $\phi$, but may have running time
significantly less.

We require a finer understanding of the way in which paths may cross.
We denote the sequence of crossings between $P$ and $Q$ by $P \otimes
Q$ with ordering inherited from $Q$. Although the ordering of
$P\otimes Q$ and $Q\otimes P$ may not be the same, we have that
$|P\otimes Q| = |Q\otimes P|$. 
While we only need part 2 of the following theorem, part 1 is used
within the proof for part 2 and may be of
independent interest.

\begin{theorem}[Leftmost Crossings]\label{thm:spiral}
  Let $P$ be the leftmost residual $s$-to-$t$ path \wrt $c_{{\bf
      c}^\circlearrowright}$ and let $Q$ be an $s$-to-$t$ path such
  that $\rev{Q}$ is residual \wrt $c_{{\bf c}^\circlearrowright}$. Then:
  \begin{enumerate}
  \item \label{pt1} The order of crossings is the same along both $P$ and $Q$.
    That is, either $X = Y$ or $X = \rev{Y}$ where $X$ and $Y$ are the $i^{th}$
    crossing in $P \otimes Q$ and $Q \otimes P$, respectively. 
  \item \label{pt2} $P$ crosses $Q$ from right to left at $X$ for all $X \in P \otimes Q$.
  \end{enumerate}
\end{theorem}

\begin{proof}
  If $|P\otimes Q| = 0$, the theorem is trivially true. 

  Let $P \otimes Q = \{X_1, X_2, \ldots, X_{|P \otimes Q|}\}$ and define
  $X_0 = s, X_{|P \otimes Q|+1} = t$. Likewise let $Q \otimes P =
  \{Y_1, Y_2, \ldots, Y_{|Q \otimes P|}\}$ and define
  $Y_0 = s, Y_{|Q \otimes P|+1} = t$.  For a contradiction to
  Part~\ref{pt1}, let $i$ be the smallest index
  such that  $X_i \notin \{ Y_i, \rev{Y_i}\}$.  Let $j$ be the index
  such that $Y_j \in \{X_{i+1},\rev{X_{i+1}}\}$.  Then $j \geq i$ by
  choice of $i$.

  Let $x_i$ be any vertex in $X_i$.  Since $P[x_{i-1},x_{i+1}]$ does
  not cross $Q$ at $x_i$, $P[x_{i-1},x_{i+1}]$ does not cross
  $Q[x_{i-1},x_{i+1}]$. Since $P$ and $\rev{Q}$ are residual,
  $C_1 = P[x_{i-1},x_{i+1}] \circ \rev{Q[x_{i-1},x_{i+1}]}$ is a simple
  counterclockwise cycle.

  Since there are no crossings in $Q[x_i,x_{i+1}]$, $C_2 = Q[x_i,x_{i+1}]
  \circ P[x_{i+1},x_i]$ is a simple cycle.  Since $P$ is leftmost
  residual, $P[x_{i+1},x_i]$ is left of $\rev{Q[x_i,x_{i+1}]}$ and
  $C_2$ is clockwise.

  Since $P$ and $Q$ are simple, $C_1$ and $C_2$ do not cross.
  Therefore it must be the case that either $C_1$ is enclosed by $C_2$
  or vice versa. See Fig.~\ref{fig:spiral}.

  \begin{description}
  \item[$C_2$ is enclosed by $C_1$] Since $P$ crosses $Q$ at
    $X_{i+1}$, $Q[x_{i+1},\cdot]$ must have a subpath in the strict
    interior of $C_1$.  Then, a maximal such subpath forms a
    counterclockwise cycle with a subpath of $P$, contradicting that
    $P$ is a leftmost residual path.
  \item[$C_1$ is enclosed by $C_2$] Since $P$ crosses $Q$ at $X_i$,
    $P[x_i,\cdot]$ must enter the strict interior of $C_1$.  Since
    $P[x_i,\cdot]$ does not cross $Q[x_i,x_{i+1}]$, $P[x_i,\cdot]$ is
    entirely enclosed by $C_1$, contradicting that $t$ is on the
    infinite face.
  \end{description}
  
  This proves part~\ref{pt1} of the theorem.
  Since $Q[x_i,\cdot]$ does not
  cross $P[x_i,x_{i+1}]$,  $Q[x_i,\cdot]$ does not enter the cycle
  $P[x_i,x_{i+1}]\circ \rev{Q[x_i,x_{i+1}]}$. Since $P[x_i,x_{i+1}]$
  is right of $Q[x_i,x_{i+1}]$, it follows that $P$ enters $Q$ from
  the right at $x_{i+1}$. Part~\ref{pt2} follows.\qed
\end{proof}

\begin{figure}[ht]
  \centering
  \subfigure[]
  {
    \includegraphics[scale=0.4]{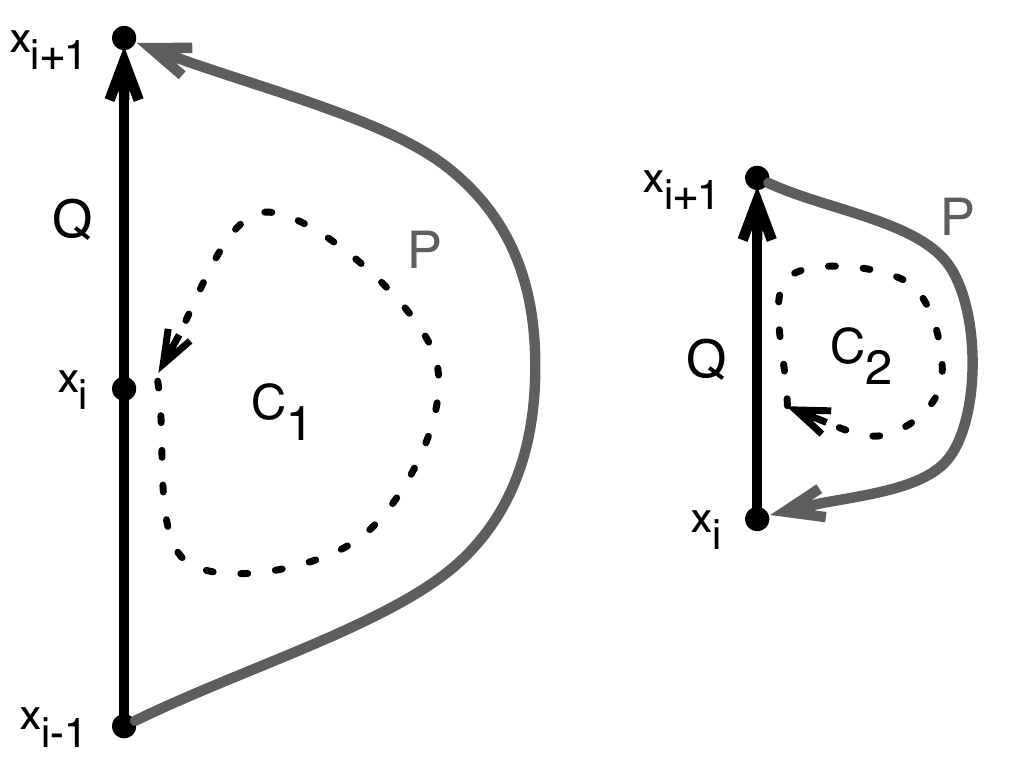}
    \label{fig:spiral_cy}
  }
  \subfigure[]
  {
    \includegraphics[scale=0.4]{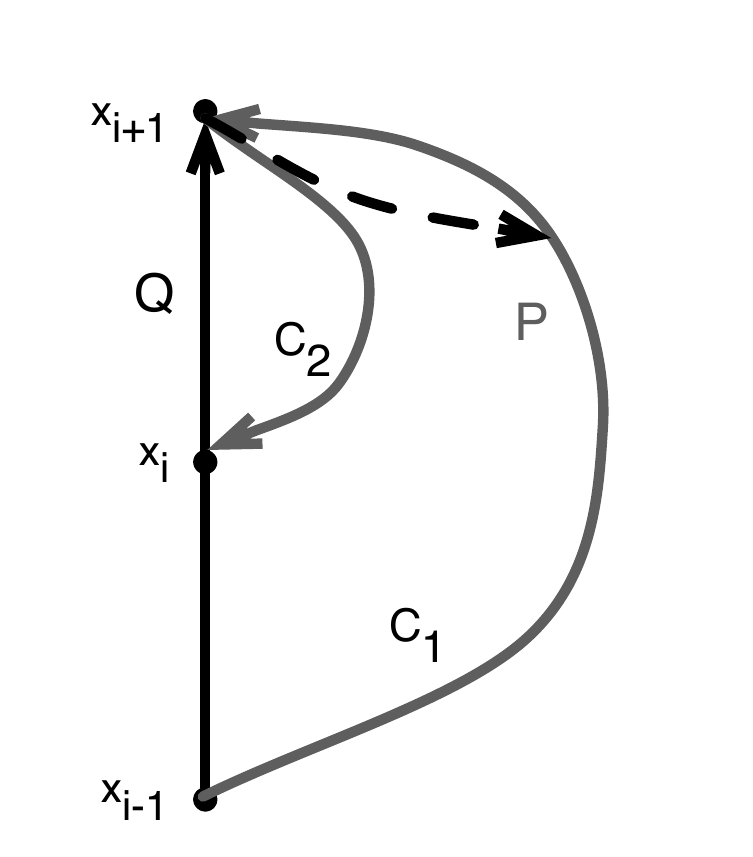}
    \label{fig:spiral1}
  }
  \subfigure[]
  {
    \includegraphics[scale=0.4]{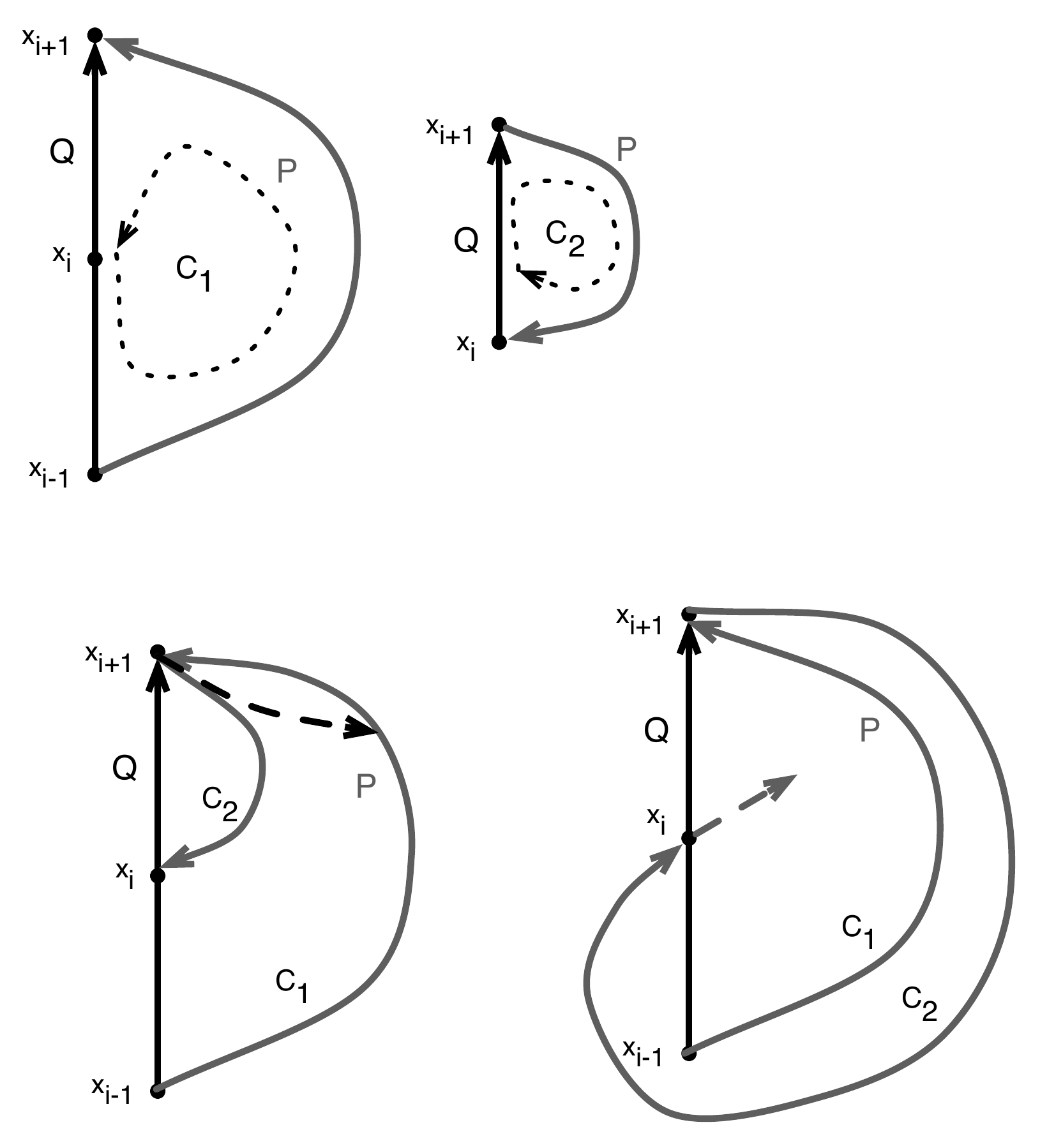}
    \label{fig:spiral2}
  }
  \caption{(a) Cycles $C_1$ and $C_2$ used in the proof of
    Theorem~\ref{thm:spiral}. (b) Case 1: $C_2$ is enclosed
    by $C_1$. (c) Case 2: $C_1$ is enclosed by $C_2$}
  \label{fig:spiral}
\end{figure}

\subsection{An adaptive algorithm}

In our new algorithm {\sc MaxAdaptiveFlow} we find a sequence of capacities
$c_0, c_1, \cdots$ where $c_i$ are the residual capacites of $c_{i-1}$
with respect to some flow.  We refer to this flow as the flow that
takes $c_{i-1}$ to $c_i$.

\begin{tabbing}
  {\sc MaxAdaptiveFlow} ($G$, $s$, $t$, $\cvec$)\\  
  \qquad let $\cvec_0$ be the residual capacities resulting from
  saturating the \cw\\
  \qquad \qquad \qquad residual cycles of $G$ \wrt $c$ \\
  \qquad $i := 0$ \\
  \qquad  while there is a residual \xty{s}{t} path in $G$ \wrt
  $\cvec_i$\\
  \qquad \qquad  let $A_i$ be the leftmost of these paths \\
  \qquad \qquad  let $\cvec_{i+1}$ be the residual capacities resulting from
  saturating the \\
  \qquad \qquad \qquad \qquad leftmost $st$-flow in $G \snip A_i$ \wrt
  $c_i$\\
  \qquad \qquad $i := i+1$\\
 \qquad  return the flow defined by $\fvec[d] = \cvec_i[d] - \cvec[d] \ \forall \text{ darts } d$.\\
\end{tabbing}

Let $\rho$ be the number of iterations of {\sc MaxAdaptiveFlow}.
$A_i$ is an \xty{s}{t} path, $s$ and $t$ are adjacent to $f_{\infty}$ in $G
\snip A_i$, and the leftmost $st$-flow in $G \snip A_i$ reduces to
a single shortest-path computation in the dual graph.  Therefore {\sc
  MaxAdaptiveFlow} can be implemented with $\rho+1$ shortest-path
computations (finding the initial circulation is an additional
shortest-path computation).  {\sc MaxAdaptiveFlow} is correct as it
does not complete until there is no residual $s$-to-$t$ path.  We will
bound $\rho$ by $2\phi+1$, giving:

\begin{theorem}\label{thm:faces}
  {\sc MaxAdaptiveFlow} can be implemented using at most
  $2\phi+2$ shortest-path computations.
\end{theorem}

We bound the number of iterations of {\sc MaxAdaptiveFlow} by showing
that the paths $A_0, A_1, \ldots$ first monotonically decrease and
then increase in the number of times they cross $\Pi$, the smallest
set of faces that a curve from $s$ to $t$ drawn on the plane in which
$G$ is embedded must cross.  Notice that $A_i$ must cross $A_{i-1}$ at
least once: since $A_i$ is residual in $G$ \wrt $c_i$, $A_i$ cannot be
a path in $G\snip A_{i-1}$ (for otherwise it would have been augmented
in iteration $i-1$).  We show that {\sc MaxAdaptiveFlow} maintains as
an invariant the absence of clockwise residual cycles; this will imply
that these crossings can only be from right to left (by
Theorem~\ref{thm:spiral}). We relate the crossings between $A_i$ and
$A_{i-1}$ to $A_i$ and $\Pi$ by viewing these paths in the covering
graph of $G$ (as defined in Section~\ref{sec:finite-cover}).

We could, in fact, show that {\sc MaxAdaptiveFlow} implements the
leftmost-augmenting-path algorithm of Borradaile and Klein~\cite{BK09}
insofar as each iteration of {\sc MaxAdaptiveFlow} corresponds to a
consecutive batch of augmentations of the leftmost-augmenting-path
algorithm.  This correspondence would imply the following invariant.
However, it is less cumbersome to prove this invariant directly.

\begin{invariant}\label{inv:no-cw-cycles}
$G$ has no clockwise residual cycles \wrt $\cvec_i$.
\end{invariant}

\begin{proof}[by induction]

  If $i=0$ the invariant holds trivially by definition of $\cvec_0$. For a
  contradiction, assume that $\cvec_j$ is clockwise acyclic, but
  $\cvec_{j+1}$ is not. Let $C$ be a clockwise cycle in $G$ that is
  residual \wrt $\cvec_{j+1}$. Since $G\snip A_j$ is clockwise
  non-residual \wrt $\cvec_{j+1}$, $C$ must use a dart $d$ of
  $\leftd{A_j}$ that enters $A_j$. By the inductive
  hypothesis, $C$ is not residual \wrt $\cvec_j$; let $a$ be its last
  non-residual arc \wrt $\cvec_j$. Let $F$ be an \xty{s}{t} path in the
  flow that takes $\cvec_j$ to $\cvec_{j+1}$ and uses $\rev{a}$ and let $x$
  be the first vertex of $F$ on $C$ after $\head(a)$. Then: $F[\cdot,
  x]\circ C[x, \head(d)]$ is residual \wrt $\cvec_j$ and, since $F$ does not cross
  $A_j$, $F[\cdot, x]\circ C[x, \head(d)]\circ \rev{A_j}[head(d),s]$
  is a clockwise cycle, which contradicts $A_j$ being leftmost
  residual \wrt $\cvec_j$. \qed
\end{proof}

As a consequence of there being no clockwise residual cycles, we show
that the paths $A_0, A_1, A_2, \ldots$ move from left to right.

\begin{lemma}\label{lem:left-to-right}
$A_i$ is left of $A_{i-1}$. $A_i$ crosses $A_{i-1}$ at least once and
only from right to left.
\end{lemma}

\begin{proof}
   First we observe that $\rev{A_{i-1}}$ is residual \wrt $\cvec_i$.  We find a
  leftmost maximum flow in $G \snip A_{i-1}$, an $st$-planar graph, and 
  since $A_{i-1}$ is residual \wrt $\cvec_{i-1}$, the flow we find is
  non-trivial.  The leftmost of these flow paths must indeed by $A_{i-1}$.

  $A_i$ crosses $A_{i-1}$.  Then, since $A_i$ is leftmost residual, we
  refer to the properties guaranteed by Theorem~\ref{thm:spiral},
  proving the second part of the lemma.

  Let $A_i \otimes A_{i-1} = \{X_1, X_2, \ldots, X_{|A_i \otimes
    A_{i-1}|}\}$ and define $X_0 = s, X_{|A_i \otimes A_{i-1}|+1} =
  t$. Let $x_j$ be any vertex in $X_j$. Consider the subpaths
  $A_i[x_j,x_{j+1}]$ and $A_{i-1}[x_j,x_{j+1}]$. The cycle
  $A_i[x_j,x_{j+1}]\circ \rev{A_{i-1}[x_j,x_{j+1}]}$ is residual, and
  by Invariant~\ref{inv:no-cw-cycles} cannot be clockwise. Therefore
  $A_i[x_j,x_{j+1}]$ is left of $A_{i-1}[x_j,x_{j+1}]$. \qed
\end{proof}

\begin{proof}[of Theorem~\ref{thm:faces}]
  Since $A_i$ crosses $A_{i-1}$ at least once and from right to left.
  Therefore there exists a subpath of $A_i$ such that $Y$ leaves
  $A_{i-1}$ from the left and enters $A_{i-1}$ from the right and
  there is no subpath of $A_i$ that leaves $A_{i-1}$ from the right
  and enters $A_{i-1}$ from the left.  Then by
  Lemma~\ref{lem:mapping}, $A_{i}^0$ must contain at least one
  subpath from $A_{i-1}^{j+1}$ to $A_{i-1}^j$ for some $j$ (and no
  $A_{i-1}^{j}$-to-$A_{i-1}^{j+1}$ subpaths, for any $j$).  Therefore
  we make progress in the following sense:
  \begin{quote}
    $A_{i}^0$ starts at a source strictly to the right of the source that $A_{i-1}^0$ starts at in $\cal G$.
  \end{quote}
  From the Pigeonhole Lemma~\ref{lem:pigeonhole}, each $A_i^0$ can go
  through at most $\phi$ copies of $G$ in ${\cal G}$.  Therefore
  $A^0_i$ must start at a source within $\phi$ isomorphic copies of $G
  \snip \Pi$ of the source that $\Pi^0$ starts at (when $\pi$ is an
  artificial path as in the proof of the Pigeonhole Lemma).  It
  follows that the number of iterations of {\sc MaxAdaptiveFlow} is at
  most $2\phi+1$. \qed
\end{proof}

Note that the bound in the above argument is not tight.  The progress could be much greater in each
iteration of {\sc MaxAdaptiveFlow}.

\end{document}